\iffalse
\documentclass[a4paper,UKenglish]{lipics-v2019}
\def\versy{conf}
\def\webpage{web}
\def\conf{conf}
\usepackage{tikz}
\else
\documentclass[a4paper,UKenglish]{lipicsown}
\def\versy{web}
\def\webpage{web}
\def\conf{conf}
\fi
\usepackage{amsfonts}
\usepackage{amssymb}
\usepackage{algorithm}
\usepackage{algorithmic}
\usepackage{float}
\usepackage{enumerate}
\usepackage{pgfplots, tikz}
\usepackage{amsmath}
\newtheorem{propo}[theorem]{Proposition}
\newtheorem{remy}[theorem]{Remark}

\newlength\myindent 
\def\MHD{\mbox{MHD}}

\ifx\versy\webpage
\title{A Fast Exponential Time Algorithm for
Max Hamming Distance X3SAT}
\titlerunning{A Fast Algorithm for Max Hamming Distance X3SAT}
\author[2,3]{Gordon Hoi$^1$, Sanjay Jain$^2$ and Frank Stephan}
\affil[1]{School of Computing, National University of Singapore,
13 Computing Drive, Block COM1, Singapore 117417,
Republic of Singapore, \texttt{e0013185@u.nus.edu}}
\affil[2]{School of Computing, National University of Singapore,
COM1, 13 Computing Drive, Singapore 117417, Republic of Singapore,
\texttt{sanjay@comp.nus.edu.sg and fstephan@comp.nus.edu.sg}}
\affil[3]{Department of Mathematics, National University of Singapore,
10 Lower Kent Ridge Road, Block S17, Singapore 119076, Republic of Singapore;
\newline \mbox{ } \newline
S.~Jain and F.~Stephan were
supported in part in part by the Singapore Ministry of Education Academic
Research Fund Tier 2 grant MOE2016-T2-1-019 / R146-000-234-112.
Additionally S.~Jain was supported by NUS grant C252-000-087-001.}
\authorrunning{G.~Hoi, S.~Jain and F.~Stephan}
\fi
\ifx\versy\conf
\title{A Fast Exponential Time Algorithm for
Max Hamming Distance X3SAT}
\titlerunning{A Fast Algorithm for Max Hamming Distance X3SAT}
\author{Gordon Hoi}{School of Computing,
National University of Singapore, Singapore 117417, Republic of Singapore}
{e0013185@u.nus.edu}{}{}
\author{Sanjay Jain}{School of Computing, National University of
Singapore, Singapore 117417, Republic of Singapore.
S.~Jain was supported in part by the
NUS grant C-252-000-087-001 and
by Singapore Ministry of Education Academic Research Fund Tier 2 grant
MOE2016-T2-1-019 / R1467-000-234-112.}{sanjay@comp.nus.edu.sg}{}{}
\author{Frank Stephan}{Department of Mathematics, National University of
Singapore, Singapore 119076, Republic of Singapore and School of Computing,
National University of Singapore, Singapore 117417, Republic of Singapore.
F.~Stephan was supported in part by the
Singapore Ministry of Education Academic Research Fund Tier 2 grant
MOE2016-T2-1-019 / R1467-000-234-112.}{fstephan@comp.nus.edu.sg}{}{}


\authorrunning{G.~Hoi, S.~Jain and F.~Stephan}

\Copyright{Gordon Hoi, Sanjay Jain and Frank Stephan}
\fi

\date{\today}
\begin{document}
\ifx\versy\conf
\nocite{*}
\Copyright{Sanjay Jain, Gordon Hoi and Frank Stephan}
\funding{S.~Jain and F.~Stephan were
supported in part in part by the Singapore Ministry of Education Academic
Research Fund Tier 2 grant MOE2016-T2-1-019 / R146-000-234-112.
Additionally S.~Jain was supported in part by NUS grant C252-000-087-001.}
\fi
\maketitle

\begin{abstract}
\noindent
X3SAT is the problem of whether one can satisfy a given set of clauses
with up to three literals such that in every clause, exactly one literal
is true and the others are false. A related question is to determine
the maximal Hamming distance between two solutions of the instance.
Dahll\"of provided an algorithm for Maximum Hamming Distance XSAT,
which is more complicated than the same problem for X3SAT,
with a runtime of $O(1.8348^n)$;
Fu, Zhou and Yin considered Maximum Hamming Distance for
X3SAT and found for this
problem an algorithm with runtime $O(1.6760^n)$.
In this paper, we propose an algorithm in $O(1.3298^n)$
time to solve the Max Hamming Distance X3SAT problem;
the algorithm actually counts for each $k$ the number of pairs of solutions
which have Hamming Distance $k$.
\end{abstract}

\ifx\versy\conf
\newpage
\fi

\section{Introduction}  

\noindent
Given a Boolean formula $\phi$ in conjunctive normal form, the
satisfiability (SAT) problem seeks to know if there are possible truth
assignments to the variables such that $\phi$ evaluates to the
value ``True''. One na\"ive way to solve this problem is to brute-force
all possible truth assignments and see if there exist any assignment
that will evaluate $\phi$ to ``True''. Suppose that there are $n$
variables and $m$ clauses, we will take up to $O(mn)$ time to check if
every clause is satisfiable. However, since there are $2^n$ different
truth assignments, we will take a total of $O(2^n nm)$ time
\cite{FK10}. Classical
algorithms were improving on this by exploiting structural properties
of the satisfiability problem and in particular its variants. The
basic type algorithms are called DPLL algorithms
--- by the initials of the authors of the corresponding
papers \cite{DLL62,DP60} --- and the main
idea is to branch the algorithm over variables where one can, from
the formula, in each of the branchings deduce consequences which
allow to derive values of some further variables as well, so that the
overall amount of the run time can be brought down. For the analysis
of the runtime of such algorithms, we also refer to the work of
Eppstein \cite{Ep01,Ep06}, Fomin and Kratsch \cite{FK10}
and Kullmann \cite{Kul99}.

A variant of SAT is the Exact Satisfiability problem (XSAT),
where we require that the satisfying assignment has exactly 1 of the
literals to be true in each clause, while the other literals in the same clause
are assigned false.
If we have at most 3 literals per
clause with the aim of only having exactly 1 literal to be true, then
the whole problem is known as Exact 3-Satisfiability (X3SAT) and this
is the problem which we wish to study. Wahlstr\"om \cite{Wah07}
provided an X3SAT solver which runs in time $O^*(1.0984^n)$ and
subsequently there were only slight improvements; here $n$ is, as
also always below, the number of variables of the given instance
and $O^*(g(n))$ is the class of all functions $f$ bounded by some polynomial
$p(\cdot)$ (in the size of the input) times $g(n)$. 
The problems mentioned before, SAT, 3SAT and X3SAT
are all known to be NP-complete. More background information to the
above bounds can be found in the PhD theses and books of Dahll\"of
\cite{Dah06}, Gaspers \cite{Gas10} and Wahlstr\"om \cite{Wah07}.

The runtime of SAT, 3SAT, XSAT and X3SAT have been well-explored. Sometimes, 
instead of just finding a solution instance to a problem, we are interested
in finding many ``diverse'' solutions to a problem instance. Generating
``diverse'' solutions is of much importance in the real world and can
be seen in areas such as Automated Planning, Path Planning and Constraint
Programming \cite{VK16}. How does one then
measure the ``diversity'' of solutions? This combinatorial aspect
can be investigated naturally with the notion of the Hamming Distance.
Given any two satisfying assignments
to a satisfiability problem, the Hamming Distance problem seeks to find the
number of variables that differ between them. The Max Hamming Distance
problem therefore seeks to compute the maximum number of variables
that will defer between any two satisfying assignments.
If we are interested in the ``diversity'' of exact satisfying assignments,
then the problem is defined as Max Hamming Distance XSAT (X3SAT) accordingly. 
The algorithm given in this paper actually provides information about 
the number of pairs of
solutions which have Hamming distance $k$, for  $k=0,1,\ldots,n$, which could
potentially have uses in other fields such as error correction.

A number of authors have worked in these area previously as well.
Crescenzi and Rossi \cite{CR02} as well as Angelsmark and Thapper
\cite{AT04} studied the question to determine the maximum Hamming distance
of solutions of instances of certain problems.
Dahll\"of \cite{Dah05,Dah06} gave two algorithms for Max Hamming Distance
XSAT problem in $O^*(2^n)$ and an improved version in $O^*(1.8348^n)$.
The first algorithm enumerates all possible subset of all sizes while checking
that they meet certain conditions. The second algorithm uses techniques
found in DPLL algorithms. Fu, Zhou and Yin \cite{FZY12} specialised
on the X3SAT problem and provided an algorithm to determine the Max Hamming
Distance of two solutions of an X3SAT instance in time $O^*(1.676^n)$.
Recently, Hoi and Stephan \cite{HS19} gave an algorithm to solve the 
Max Hamming Distance XSAT problem in $O(1.4983^n)$. 



The main objective of this paper is to propose an algorithm in
$O(1.3298^n)$ time to solve the Max Hamming Distance X3SAT problem.
The output of the algorithm is a polynomial $p$ which gives
information about the number $a_k$ of pairs of solutions of Hamming distance
$k$, for $k=0,1,\ldots,n$. The algorithm does so by simplifying in parallel
two versions $\phi_1,\phi_2$ of the input instance and the
main novelty of this algorithm is to maintain the same structure of
$\phi_1$ and $\phi_2$ and to also hold information
about the Hamming distance of the current and resolved variables while carrying
out an DPLL style branching algorithm.

\ifx\versy\webpage
Section~\ref{sec:comparision}
compares the approach taken with other known methods.
\fi

\section{Basic Approach}

\noindent
Suppose a X3SAT formula $\phi$ over the set of $n$ variables $X$ is
given. The aim is to find the largest Hamming distance possible
between two possible value assignments $\beta_1, \beta_2$ to
the variables which are solutions of $\phi$, that is, make true
exactly one literal in each clause of $\phi$.

To this end, the algorithm presented in this paper computes a polynomial 
(called {\em HD-polynomial})
in $u$, with degree at most $n$, such that the coefficient $c_k$ of
$u^k$ gives the number of solution pairs $(\beta_1,\beta_2)$ such that
the Hamming distance between $\beta_1$ and $\beta_2$ is $k$.
The degree of this polynomial will then provide the largest Hamming
distance between any pair of solutions.

\begin{example}
We consider the formula $\phi$ = $(x_1 \vee x_2 \vee x_3) \wedge (x_1 \vee x_
4 \vee x_5)
\wedge (x_1 \vee x_6 \vee x_7) \wedge (x_2 \vee x_4 \vee \neg x_6)$.
Exhaustive search gives for this X3SAT formula the following four solutions:
\begin{center}
\begin{tabular}{|c|c|c|c|c|c|c|}
\hline
$x_1$ & $x_2$ & $x_3$ & $x_4$ & $x_5$ & $x_6$ & $x_7$ \\ \hline
$1$ & $0$ & $0$ & $0$ & $0$ & $0$ & $0$ \\ \hline
$0$ & $1$ & $0$ & $0$ & $1$ & $1$ & $0$ \\ \hline
$0$ & $0$ & $1$ & $1$ & $0$ & $1$ & $0$ \\ \hline
$0$ & $0$ & $1$ & $0$ & $1$ & $0$ & $1$ \\ \hline
\end{tabular}
\end{center}
So there are $16$ pairs of solutions among which four pairs have
Hamming distance $0$ and twelve pairs of Hamming distance $4$.
The intended output of the algorithm is the
polynomial $12 u^4 + 4 u^0$ which indicates that there are four
pairs of Hamming distance $0$ and twelve pairs of Hamming distance $4$.
\end{example}

\noindent
The reason for choosing this representation is that our algorithm often
needs to add/multiply possible partial solutions, which can
be done easily using these polynomials whenever needed.

The brute force approach would be to consider a search tree, with
four branches at the internal nodes --- $(0,0),(0,1),(1,0),(1,1)$
based on values assigned to some variable $x$ for the two possible 
solutions being compared.
If at a leaf the candidate value assignments $(\beta_1,\beta_2)$ 
formed by using the values chosen along the path from the root
are indeed both solutions for $\phi$ and their Hamming distance is $k$,
then the polynomial calculated at the leaf would be $u^k$;
if any of $(\beta_1,\beta_2)$ are not solutions then the polynomial 
calculated at the leaf would be $0$.
Then, one adds up all the polynomials at the leaves to get the result.
This exhaustive search has time complexity 
(number of leaves)${} \times poly(n,|\phi|)=4^n \times poly(n,|\phi|)$ 
for $n$ variables.

For $x \in X$ and $i,j \in \{0,1\}$,
let $q_{x,i,j}$ be $u$ if $i \neq j$ and $1$ otherwise.
The above brute force approach for computing the HD-polynomial 
would be equivalent to computing
$$\sum_{(\beta_1,\beta_2)} \prod_{x \in X} \  \ q_{x,\beta_1(x),\beta_2(x)},$$
where $(\beta_1,\beta_2)$ in the summation ranges over the pair of
solutions for the X3SAT problem $\phi$. 

However, we may not always need to do the full search as above.
We will be using a DPLL type algorithm, where we use branching
as above, and simplifications at various points to reduce the
number of leaves in the search tree.
Note that the complexity of such algorithms is 
proportional to the number of leaves, modulo a polynomial factor:
that is, complexity is 
$O(poly(n,|\phi|) \times ($number of leaves in the search tree$))=
O^*$(number of leaves in the search tree).

As an illustration we consider some examples where the problems
can be simplified. If there is a clause $(x,y)$, then 
$x = \neg y$ for any solution which satisfies the clause. Thus, $x$
and $y$'s values are linked to each other, and we only need to explore
the possibilities for $y$ and can drop the branching for $x$
(in addition one needs to do some book-keeping to make sure the difference in
the values of $y$ in two solutions also takes care of the difference 
in the values of $x$ in the two solutions; this book-keeping will be explained
below).
As another example,
if there is a clause $(x,x,z)$, then value of $x$ must be $0$ in
any solution which satisfies the clause.
Our algorithm would use several such simplifications to bring down 
the complexity of finding the largest Hamming distance. In the simplification
process, we will either fix values of some of the variables, or link
some variables as above, or branch on a variable $x$ to restrict possibilities
of other variables in clauses involving $x$ and so on (more details below).

In the process, we need to maintain that the HD-polynomial generated
is as required.
Intuitively, if we consider a polynomial calculated at any node as
the sum of the values of the polynomials in the leaves which are
its descendant, then the value of the polynomial
calculated at the root of the search tree gives  the HD-polynomial we want.
For this purpose, we will keep track of polynomials named $p_{main}$
and $p_{x,i,j}$, which start with $p_{main}$ being $1$,
and polynomials $p_{x,i,j}=q_{x,i,j}$, for $x \in X, i,j \in \{0,1\}$
(here $q_{x,i,j}$ is $u$ for $i \neq j$, and $1$ otherwise).
If there is no simplification done, then
at the leaves, the polynomial $p_{main}$ will become
the product of $p_{x,i,j}$, $x \in X$, for the values $(i,j)$ taken by
$x$ for the two solutions in that branch.
When doing simplification via linking of variables, or assigning truth value to
some variables, etc. we will update these polynomials, so as to maintain
that the polynomial calculated at the root using above
method is the HD-polynomial we need.
More details on this updating would be given in the following section.

\section{Algorithm for Computing HD-polynomial}

\noindent
In this section we describe the algorithm for finding the HD-polynomial
for any X3SAT formula $\phi$.
Note that we consider clause $(x,y,z)$ to be same as $(y,x,z)$, that
is order of the literals in the clause does not matter.
We start with some definitions.

Notation: For a formula $\phi$ with variable $x$, we use the
notation $\phi[x=i]$ to denote the formula obtained by replacing all
occurence of $x$ in $\phi$ by $i$. Similarly, for a set $P$ containing
values/definitions of some parameters, including $p_1,p_2$, we use
$P[p_1=f,p_2=g]$ to denote the modification of $p_1$ to $f$, $p_2$ to 
$g$ (and rest of the parameters remaining the same).

\begin{definition}
Fix a formula $\phi$:
\begin{enumerate}[(a)]
\item For a literal / variable $x$, $x'$ and $x''$ and other primed
versions are either $x$ or $\neg x$, i.e., they use the same variable
$x$, which may or may not be negated.

\item Two clauses $c,c'$ are called {\em neighbours} if they share a common
variable. For example, $(x,y,z)$ and $(\neg x,w,r)$ are neighbours.

\item 
Two clauses are called {\em similar} if one of them can be obtained from the 
other
just by negating some of the literals. They are called {\em dissimilar} if
they are not similar.
For example, $(x, y)$ is similar to $(x, \neg y)$,
$(1, x,y)$ is similar to $(0, \neg x, y)$, 
$(x, z)$ is dissimilar to $(x, y)$ and
$(x, \neg x, z)$ is dissimilar to $(x, z, \neg z)$.

\item Two X3SAT formulas have the {\em same structure} if they have the same
number of clauses and there is a 1--1 mapping between these clauses 
such that the mapping maps a clause to a similar clause.

\item A set of clauses $C$ is called {\em isolated} (in $\phi$), if 
none of the clauses in $C$ is a neighbour of any clause in $\phi$ which is
not in $C$.

\item 
A set $I$ of variables is {\em semisolated in $\phi$ by $J$} if all the clauses
in $\phi$ either contain only variables from $I \cup J$, or do not
contain any variable from $I$.
We will be using such $I$ and $J$ for
$|I| \leq 10$ and $|J| \leq 3$ only to simplify some cases.

\item We say that $x$ is {\em linked} to $y$, if 
we can derive that $x=y$ (respectively, $x = \neg y$) in any possible
solution using constantly many clauses of the X3SAT formula $\phi$
as considered in our case analysis (a constant bound of $20$ is enough).
In this case we say that value $i$ of $x$ is linked
to value $i$ of $y$ (value $i$ of $x$ is linked to value $1-i$ of $y$
respectively).
\end{enumerate}
\end{definition}

\begin{definition}[see Monien and Preis \cite{MP06}]
Suppose $G=(V,E)$ is a simple undirected graph. A {\em balanced bisection} is
a mapping $\pi:V \to \{0,1\}$ such that, 
for $V_i=\{v: \pi(v)=i\}$, $|V_0|$ and $|V_1|$ differ by at most one.
Let $cut(\pi)=|\{(v,w): v \in V_0, w \in V_1\}|$. The bisection
width of $G$ is the smallest $cut(\cdot)$ that can be obtained for a balanced
bisection.
\end{definition}


\noindent
Suppose $\phi$ is the original X3SAT formula given over $n$ variable
set $X$.
Our main (recursive) algorithm is $\MHD(\phi_1,\phi_2,s_1,s_2,V,P)$,
where $\phi_1,\phi_2$ are formulas with the same structure over
variable set $V
\subseteq X$,
$s_1,s_2$ are some value assignments to variables from $X$ 
and $P$ is a collection of polynomials (over $u$) for $p_{main}$ and
$p_{x,i,j}$, $x \in X$, $i,j \in \{0,1\}$.
Intuitively, $p_{main}$ represents the portion of the polynomial
which is formed using variables which have already been fixed (or implied)
based on earlier branching decisions.

Initially, algorithm starts with $\MHD(\phi_1=\phi,\phi_2=\phi,V=X,
s_1=\emptyset,s_2=\emptyset,P)$,
where $\phi$ is the original formula given for which we want to find
the Hamming distance, $X$ is the set of variables for $\phi$,
$s_1,s_2$ are empty value assignments, $p_{main}=1$, 
$p_{x,i,j}=q_{x,i,j}$.

Intuitively, the function
$\MHD(\phi_1,\phi_2,s_1,s_2,V,P)$ returns
the polynomial $p_{main} \times$ 
$\sum_{(\beta_1,\beta_2)}$ $\prod_{x \in V}[p_{x,\beta_1(x),\beta_2(x)}]$,
where $\beta_1,\beta_2$ range over value assignments to variables in $V$
which are satisfying for the formula $\phi_1$ and $\phi_2$ respectively,
and which are consistent with the value assignment in $s_1,s_2$, if any,
respectively.
Thus, if we consider the search tree, then the node representing
$\MHD(\phi_1,\phi_2,s_1,s_2,V,P)$ basically represents the polynomial
formed 
$$\sum_{(\beta_1,\beta_2)} \prod_{x \in X} \  \ q_{x,\beta_1(x),\beta_2(x)},$$
where $(\beta_1,\beta_2)$ in the summation ranges over the pair of
solutions for the X3SAT problem $\phi$, consistent with the 
choices taken for the branching variables in the path from the root to the
node. Over the course of the algorithm, the following steps will be done:  

\begin{enumerate}[(a)]
\item using polynomial amount of work (in size of $\phi$)
branch over some variable or group of variables. That is, if 
we branch over variable $x$, we consider
all possible values
for $x$ in $\{0,1\}$ for $\phi_1, \phi_2$ (consistent with
$s_1(x),s_2(x)$ respectively), and then evaluate
the corresponding subproblems: note that $\MHD(\phi_1,\phi_2,s_1,s_2,V,P)$
would be the sum of the answers returned by (upto) four subproblems created 
as above: where in the subproblem for $x$ 
being fixed to $(i,j)$ in $(\phi_1,\phi_2)$ respectively,
$p_{main}$ gets multiplied by $p_{x,i,j}$ and $x$ is dropped from
$V$.

\item simplify the problem,
using polynomial (in size of $\phi$) amount of work, 
to $\MHD(\phi_1',\phi_2',s_1',\linebreak[3] s_2',V',P')$, where 
we reduce the number of variables in $V$ or the number of clauses
in $\phi_1',\phi_2'$.
\end{enumerate}

\noindent
Note that all our branching/simplication rules will maintain the
correctness of calculation of $\MHD(\ldots)$ as described above.

Thus, the overall complexity of the algorithm is 
$O(poly(n,|\phi|) \times [$number of leaves in search tree$])$.
In the analysis below thus, whenever
branching occurs, reducing the number of variables from $n$
to $n-r_1,n-r_2,\ldots,n-r_k$ in various branches, then
we give a corresponding $\alpha_0$ such that for
all $\alpha \geq \alpha_0$, $\alpha^n \geq 
\alpha^{n-r_1}+ \alpha^{n-r_2}+\ldots \alpha^{n-r_k}$. 
Having these $\alpha_0$'s for each of the cases below would thus give
us that the overall complexity of the algorithm is at most
$O(poly(n,|\phi|)*\alpha_1^n)$, for any $\alpha_1$ larger than any of the
$\alpha_0$'s used in the cases.

All of our modifications done via case analysis below
would convert similar clauses to similar clauses.
Thus, if one starts with $\phi_1=\phi_2$, then as we proceed
with the modifications below, the corresponding clauses in the modified 
$\phi_1,\phi_2$ would remain similar (or both dropped) in 
the new (sub)problems created. Thus, $\phi_1,\phi_2$ will always have
the same structure.

Our algorithm/analysis is based on two main cases. Initially, first
case is
applied until it can no longer be applied. Then, Case 2 applies, repeatedly
to solve the problem (Case 2 will use simplifications as in
Case 1.(i) to (iv), but no branching from Case 1).
The basic outline of the algorithm is
given below, followed by the detailed case analysis.

\medskip
\noindent
{\bf Algorithm Max Hamming Distance X3SAT}: $\MHD(\phi_1,\phi_2,V,s_1,s_2,P)$
\begin{algorithmic}
\STATE {\bf Output:} The polynomial $p_{main} \times$ 
$\sum_{(\beta_1,\beta_2)}$ $\prod_{x \in V}[p_{x,\beta_1(x),\beta_2(x)}]$,
where $\beta_1,\beta_2$ range over value assignments to variables in $V$
which are satisfying for the formula $\phi_1$ and $\phi_2$ respectively,
and which are consistent with the value assignment in $s_1,s_2$, if any,
respectively.

Note: As $\phi_1,\phi_2$ have the same structure, the statements
   below about two clauses being neighbours, or involving $k$-variables
   (and other similar questions) have the same answer for
   both $\phi_1,\phi_2$.

  \IF{(some clause cannot be satisfied (for example $(0,0,0)$ or 
   $(1,x,\neg x)$) in $\phi_1$ or $\phi_2$)}
    \STATE return $0$. This is Case 1.(i).
  \ELSIF{(for some variable $x \in V$, $s_1(x)$ and $s_2(x)$ are both
    defined) or ($x$ does not appear in any of the clauses)}
    \STATE
     return $\MHD(\phi_1,\phi_2,
       s_1,s_2,V-\{x\},P[p_{main}=p_{main}\times 
        (\sum_{i,j}p_{x,i,j})])$, where summation is over pairs of $(i,j)$
     which are consistent with $(s_1(x),s_2(x))$ (if defined).
     This is Case 1.(ii).
  \ELSIF{(some clause contains at most two different variables in its
      literals)}
    \STATE simplify ($\phi_1$, $\phi_2$) according to Case 1.(iii)
         and return the answer from the updated MHD problem.
  \ELSIF{(there are two clauses sharing exactly 2 common variables)}
    \STATE simplify ($\phi_1$, $\phi_2$) according to Case 1.(iv)
         and return the answer from the updated MHD problem.
  \ELSIF{(there is a variable appearing in at least 4 dissimilar clauses)}
     \STATE branch on this variable and do follow-up linking of the
     variables according to Case 1.(v), return the sum of the
    answers obtained from the subproblems.
  \ELSIF{(there is a clause with at least four dissimilar neighbours and
    there is a small set $I$ of variables which are semiisolated by 
    a small set $J$ of variables and conditions prescribed in Case 1.(vi) below
    hold; we use this only if $|I|\leq 10, |J| \leq 3$)}
    \STATE branch on all variables except one in $J$ and simplify 
       according to Case 1.(vi) and return the sum of the answers
      obtained from the subproblems.
\ELSIF{(there is a clause with at least 4 dissimilar neighbouring clauses)}
     \STATE branch on upto three variables and do follow-up linking
        according to Case 1.(vii) and return the sum of the answers from the 
        subproblems.
 \ELSE
  \STATE In this case all the clauses have at most three dissimilar neighbours,
      no variable appears in more than 3 dissimilar clauses and each clause
      has exactly three variables and no two dissimilar clauses share
      two or more variables.
  \STATE As described in Case 2 below, one can branch on some variables 
    and after simplification, have two sets of clauses in $\phi_1$ ($\phi_2$)
     which have no common variables. Furthermore, as the clauses
     do not satisfy the preconditions for Case 1, they again fall in
     Case 2, and we can repeatedly branch/simplify the formulas until
     the number of variables/clauses become small enough to use brute force.
  \ENDIF
\end{algorithmic}

\subsection{Case 1} \label{se:case1}

\noindent
This case applies when either some clause is not satisfiable irrespective of
the values of the variables (case (i)) or some variable in $V$'s value 
has already
been determined for both $\phi_1, \phi_2$ (case (ii))
or some clauses in $\phi_1$ (and thus $\phi_2$)
use only one or two variables (case (iii)),
or two dissimilar clauses have two common variables (case (iv)), 
or some variable appears in four dissimilar clauses (case (v))
or some clause has four dissimilar clauses as neighbours (which is 
divided into two subcases (vi) and (vii) below for ease of analysis).

The subcases here are in order of priority.
So (i) has higher priority than (ii) and (ii) has higher priority than 
(iii) and so on.

\begin{enumerate}[(i)]
\item If there is a clause which cannot be satisfied 
(for example the clauses $(0, 0, 0)$ or $(1, 1, x)$ or
$(1, x, \neg x)$) whatever
the assignment of values to the variables consistent with $s_1,s_2$
in either $\phi_1$ or $\phi_2$ respectively,
then $\MHD(\phi_1,\phi_2,s_1,s_2,V,P)=0$.

\item If a variable $x \in V$ is determined in both 
$\phi_1, \phi_2$ (i.e., $s_1(x)$ and $s_2(x)$ are defined),
or variable $x$ does not appear in any of the clauses, then
do the simplification:
update $p_{main}$ to $p_{main}\times 
(\sum_{i,j} p_{x,i,j})$, where
$i,j$ range over value assignments to $x$ in $\phi_1, \phi_2$ which
are consistent with $(s_1(x),s_2(x))$ (if defined) respectively.
That is, answer returned in this case is 
$\MHD(\phi_1[x=s_1(x)],\phi_2[x=s_2(x)],
s_1,s_2,V-\{x\},P[p_{main}=p_{main}
\times (\sum_{i,j} p_{x,i,j}))$, where the summation is over $i,j$
consistent with $s_1(x),s_2(x)$, if defined.

\item 
If there is a clause which contains only one variable.
Then, either the value of the variable is determined (for example
when the clause is of the form $(x, \neg x, \neg x)$ or $(x)$,
for some literal $x$, which is satisfiable only via $x=1$),
or the clause is unsatisfiable (for example when it is of the form 
$(x, x)$ or $(x, x, x)$ --- in which case we have
that $\MHD(\phi_1,\phi_2,s_1,s_2,V,P)=0$) or it does not
matter what the value of the variable is for the clause to be satisfied
(for example, when the clause is $(x, \neg x)$).
Thus, we can drop the clause and note down the value of the variable
in the corresponding $s_i$ if it is determined (if this
is in conflict with the variable having been earlier determined
in $s_i$, then $\MHD(\phi_1,\phi_2,\ldots)=0$). Note that $x$ may be
determined in only one of $\phi_1,\phi_2$, thus we do not update the
$x$ appearing
in any of the remaining clauses of $\phi_1,\phi_2$ to maintain that the
clauses of $\phi_1,\phi_2$ are similar.

If there is a clause which contains literals involving exactly two variables, 
$x$ and $y$, then $x$ and $y$ can be linked, 
either as $x=y$ or $x=\neg y$, as we must have exactly one literal
in the clause which is true for any satisfying assignment. Thus, we can replace 
all usage of $y$ by $x$ (or $\neg x$) in both $\phi_1, \phi_2$,
drop the variable $y$ from $V$ and
correspondingly, update, for $i,j \in \{0,1\}$,
$p_{x,i,j}$ to $p_{x,i,j} \times p_{y,i',j'}$,
based on the linking of values $i$ for $x$ in $\phi_1$
($j$ for $x$ in $\phi_2$ respectively) to value $i'$ for $y$ in $\phi_1$
($j'$ for $y$ in $\phi_2$ respectively).
Here, in case value of $y$ is determined in $s_1,s_2$, then
the value of $x$ is correspondingly determined --- and in case it
is in conflict with an earlier determination then $\MHD(\phi_1,\phi_2,\ldots)$
is $0$.

So for below assume no clause has literals involving at most two variables.

\item Two clauses share two of the three variables in the literals:

Suppose the clauses in $\phi_1$ are $(x, y, w)$ and
$(x', y', z)$, where $x,x'$ (similarly, $y,y'$) are literals over 
same variable.

If $x=x', y=y'$, then we have $w=z$;

If $x=\neg x', y= \neg y'$, then we must have $w=z=0$;

If $x=x', y= \neg y'$, then we must have $x=0$ and $w = \neg z$;
(case of $x=\neg x'$ and $y=y'$ is symmetrical).

In all the four cases, we have that $w$ is linked to $z$ and thus, $z$ can
be replaced using $w$ in both $\phi_1, \phi_2$, with corresponding update
of $p_{w,i,j}$ by $p_{w,i,j} \times p_{z,i',j'}$, where $i',j'$ are obtained 
from $i,j$
based on the linking in $\phi_1,\phi_2$ respectively.
Here, in case value of $z$ is determined in $s_1,s_2$, then
the value of $w$ is correspondingly determined --- and in case it
is in conflict with an earlier determination then $\MHD(\phi_1,\phi_2,\ldots)$
is $0$.

\item
A variable $x$ appears in at least four dissimilar clauses.

By Cases 1(iii) and 1(iv), these four clauses use, beside $x$, variables 
$(y_1,z_1)$, $(y_2,z_2)$, $(y_3,z_3)$, $(y_4,z_4)$ respectively,
which are all different from each other.
We branch based on $x$ having values (for $(\phi_1,\phi_2)$): 
$(0,0), (0,1), (1,0)$ and $(1,1)$.
Then, in each of the four clauses involving $x$, we link the remaining
$y_i$ and $z_i$.
Formulas $\phi_1, \phi_2$ and $s_1,s_2,V,P$ are correspondingly updated
(that is, $x$ is dropped from $V$, $p_{main}$ is updated to
$p_{main} \times p_{x,i,j}$ based on the branch $(i,j)$,
and the linking of the variables is done as in Case 1.(iii)).

Note that for each branch, we thus remove the variable $x$, 
and one of the other variables in each of the four clauses.
Thus we can remove a total of 5 variables for each subproblem based on
the branching for $x$.

\item 
Though technically we need this case only when some clause has
four neighbours (see case (vii) and Proposition~\ref{prop-3}),
the simplification can be done in other cases also.

There exists $(I,J)$, $I \cup J \subseteq V$,
such that $|I| \leq 10$,
$|J| \leq 3$ and $(I,J)$ is semiisolated in $\phi_1$ (and thus in
$\phi_2$ too) and one of the following cases hold.
\begin{enumerate}[\mbox{ } \ 1.]
\item $j= 1$ and $i \geq 1$: Suppose $J=\{x\}$.
In this case, we can simplify the formulas $\phi_1,\phi_2$
to remove variables from 
$I$ as follows:

 Let $W=\{$value vectors $(\beta_1,\beta_2)$ with domain
           $I \cup \{x\}: \beta_i$ is consistent with $s_i$ and
           all clauses involving variables $I \cup \{x\}$ in $\phi_i$
           are satisfied using $\beta_i\}$.

 Let $W_{i,j} = \{(\beta_1,\beta_2) \in W: \beta_1(x) = i 
   \wedge \beta_2(x) = j\}$.

 Let $p_{x,i,j} = p_{x,i,j} \times (\sum_{(\beta_1,\beta_2) \in W_{i,j}}
                \prod_{v \in I} p_{v,\beta_1(v),\beta_2(v)})$.

 Let $V = V - I$.

 Remove from $\phi_1$ and $\phi_2$ all clauses containing variables
           found in $I$.
If $x$ occurs in any clause after the modification, then
answer returned is $\MHD(\phi_1,\phi_2,s_1,s_2,V,P)$, where the parameters
 are modified as above.

IF {$x$ does not occur in any clause after above modification},
then, let
        $p_{main} = p_{main} \times \sum_{i,j} p_{x,i,j}$, where
     summation is over values $(i,j)$ for $x$ which are consistent
     with $(s_1(x),s_2(x))$ if defined.
        $V = V-I-\{x\}$ and the answer returned is
$\MHD(\phi_1,\phi_2,s_1,s_2,V,P)$, where the parameters
 are modified as above.

Here note that $j=0$ case can be similarly handled.
\item $J=\{w,x\}$ and $i \geq 3$, where $x$ appears in some clause $C$
involving a variable not in $I \cup J$.

In this case, we will branch on $x$ and then using the technique of
(vi).1 remove variables from $I$ and then also link the two variables 
different from $x$ in $C$. 
That is, for each $(i,j) \in \{(0,0),(0,1),(1,0),(1,1)\}$, that
is consistent with $(s_1(x),s_2(x))$
subproblem $(\phi_{1,i,j}, \phi_{2,i,j},s_{1,i,j},s_{2,i,j}, V_{i,j},P_{i,j})$
is formed as follows:
\begin{enumerate}[\rm (a)]
\item Set values of $x$ in $\phi_1$ and $\phi_2$ as $i$ and $j$ respectively,
     updating correspondingly $p_{main}$ to
     $p_{main} \times p_{x,i,j}$ and drop $x$ from the variables $V$.
\item Eliminate $I$ from the subproblem by using the method in (vi).1 (as
     $w$ is the only element of corresponding $J$ in the subproblem).
\item Link the two variables in the clause $C$ which are different from
   $x$.
\end{enumerate}
The answer returned by MHD is the sum of the answers of each of the 
four subproblems.

 Note that in each of the four (or less) subproblems,
besides $x$ and members of $I$,
 one linked variable in $C$ is removed. Thus, in total at least 5 variables 
 get eliminated in each subproblem.

\item $j=3$ and $i \geq 4$ and there is a clause which contains at
least two variables $v,w$ from $J$ and another variable from $I$
(say the clause is $(v',w',e)$), where $v',w'$ are literals
involving $v,w$); furthermore
$v,w$ appear in clauses involving variables not from $I$:
In this case we will branch on the variables $v,w,e$ (consistent
with assignments in $s_1,s_2$ to these variables if any),
and simplify each of the subproblems in a way similar to (vi).2 above.
Note that exactly one of $(v',w',e')$ is $1$: giving 9 branches
based on the three choice for each of $\phi_1$ and $\phi_2$.
The answer returned by MHD is the sum of the answers of each of the (upto)
nine subproblems.

Note that apart from the 4 elements of $I$ and $v,w$, for the clauses
using variables not from $I$, we have two clauses involving
$v$ and $w$. The other variables in each of these clauses can be linked up.
Thus, in total for each of the subproblems
at least $8$ variables are eliminated.

\end{enumerate}
\item There exists a clause with 4 dissimilar neighbours and none 
of the above cases apply.

Proposition~\ref{prop-3} below argues that
there is a clause $(x,y,z)$ (in $\phi_1$ and thus in $\phi_2$)
with at least four neighbours so that
further clauses according to one of the following five situations
exist (up to renaming of variables):
\begin{enumerate}[\mbox{ } \ 1.]
\item $(x',a,b)$, $(x'',c,d)$, $(y',a',c')$, $(y'',e,\cdot)$;
\item $(x',a,b)$, $(x'',c,d)$, $(y',e,\cdot)$, $(y'',f,\cdot)$;
\item $(x',a,b)$, $(x'',c,d)$, $(y',a',c')$, $(z',e,\cdot)$;
\item $(x',a,b)$, $(x'',c,d)$, $(y',e,\cdot)$, $(z',f,\cdot)$;
\item $(x',a,b)$, $(x'',c,d)$, $(y',a',e)$, $(z',c',e')$.
\end{enumerate}
where primed versions of the literals use the same variable as
unprimed version (though they maybe negated) and $a,b,c,d,e,f,x,y,z$
are literals involving distinct variables.
Here $\cdot$ stand for literals involving variables different from 
$x,y,z$, where it does not matter what these variables are, 
as long as they do not
create a situation as in cases 1.(i) to 1.(vi).

Suppose the clause corresponding to $(x,y,z)$ in $\phi_2$ is
$(x''',y''',z''')$. Then we branch based on
$(x,x''') =(0,0)$ or
$(x,y,z;x''',y''',z''') \in \{$
$(1,0,0;1,0,0)$,
$(1,0,0;0,1,0)$,
$(1,0,0;0,0,1)$,
$(0,1,0;1,0,0)$,
$(0,0,1;1,0,0)\}$.
That is either both of $x,x'''$ are  $0$, or at least one of
them is $1$ (as before, the branches are only used if the values
are consistent with $s_1,s_2$). 
The branch based on $x$ being $0$ in $\phi_1$ and
$x'''$ being $0$ in $\phi_2$ allows us to remove $x$ and 
three variables from linking $y$ with $z$, $a$ with $b $ and $c$ with $d$
(a total of four variables).
The branch based on the remaining 5 cases allows us to
remove $x, y, z$ and four other variables by linking the variables
other than $x,y,z$ in each of the neighbouring clause in the five possibilities
1--5 mentioned above (a total of seven variables for each of these
subproblems).
\end{enumerate}

\begin{propo} \label{prop-3}
If cases 1.(i) to 1.(vi) above do not apply and if there
is a clause with at least four dissimilar neighbours then
there is also a clause with neighbours as outlined in (vii).
\end{propo}

\begin{proof}
Below primed versions of variables denote a literal involving
the same variable --- though it may be negated version.
Given a clause $(x,y,z)$ with at least four dissimilar neighbours, 
without loss of generality assume that $x, y, z$ are not negated in
this clause (otherwise,
we can just interchange them with their negated versions).
We let $x$ denote a
variable which is in at least two further dissimilar clauses.
In the light of Cases 1.(iii), 1.(iv) not applying, these clauses have
new variables $a,b,c,d$, say
$(x',a,b)$ and $(x'',c,d)$ (again without loss of generality,
$a,b,c,d$ are not negated). In light of Case 1.(v) not applying,
variable $x$ is used in no further clause.

If two new variables $e, f$, different from $a,b,c,d,x,y,z$ appear
in some clauses involving $x,y,z$ then
there are two clauses
of the form (A) $(y',e,\cdot)$ and $(y'/z',f,\cdot)$,
or (B) $(y',e,f)$ and $(y'/z',a',c')$ (note that in case (B), both
$a,b$ (similarly, both $c,d$) cannot appear in the clause 
as case 1.(iv) did not apply). Thus, 1.(vii).2 or 1.(vii).4 (in case (A))
or 1.(vii).1 or 1.(vii).3 (in case (B)) apply.

Now, assume that at most one other variable $e$, appears in any clause
involving $x,y,z$ besides $a,b,c,d$.
Without loss of generality suppose the
third neighbour of $(x,y,z)$ was $(y',a',\cdot)$, where $\cdot$ involves
variable $c$ or $e$ (it cannot involve $b$ or $z$ as Case 1.(iv) did
not apply).
Now, if $a$ or $b$ appears in a further outside clause involving a
variable other than $x,y,z,a,b,c,d,e$, then
$(x',a,b)$ has neighbours $(x,y,z), (x'',c,d), (a',y',c'/e'), (a''/b',f,\cdot)$
and thus 1.(vii).1, 1.(vii).2, 1.(vii).3 or 1.(vii).4 apply (with interchanging
of names of $y$ with $a$ and $z$ with $b$).
If none of $a$ or $b$ appears in a further outside clause involving
a variable other than $x,y,z,a,b,c,d,e$, then 
one of the cases of 1.(vi) applies with 
$I \cup J$ being $\{x,y,z,a,b,c,d\}$ or $\{x,y,z,z,b,c,d,e\}$ (based
on whether $e$ appears with any of $x,y,z$ or not in some clause),
and $J \subseteq \{c,d,e\}$ of the variables which appear in
clauses not involving $\{x,y,z,a,b,c,d,e\}$. Here note that
in case $J=\{c,d,e\}$, then the side condition of 1.(vi).3 is satisfied
using clause $(c,d,x'')$.
\end{proof}

\subsection{Case 2}\label{se:case2}

\noindent
This case applies when all clauses have exactly three variables,
no two clauses have exactly two variables in common,
no variable appears in more than three dissimilar clauses and
dissimilar clauses have at most three dissimilar
neighbours.

As our operations on similar clauses leaves them similar, for ease
of proof writing, we will consider similar clauses in any of the formulas
as ``one'' clause when counting below.

Suppose there are $m$ dissimilar clauses involving $n$ variables.
First note that for this case, 
$m \leq 2n/3$. To see this, suppose we distribute the weight $1$ 
of each variable equally among the dissimilar clauses it belongs to. Then,
each clause may get weight $(1/3,1/2,1)$ or $(1/2,1/2,1/2)$ (or more)
based on whether the variables in the clause appear in $(2,1,0)$ or $(1,1,1)$
other clauses in the worst case.
Thus, weight on each clause is at least $3/2$, and thus there are at most
$2n/3$ dissimilar clauses.

\begin{propo}
For some $\epsilon_m$ which goes to $0$ as $m$ goes to $\infty$, the following
holds.

Suppose in $\phi_1$ (and thus $\phi_2$)
there are $n$ variables and $m$ dissimilar clauses each having three
literals involving three distinct variables,
such that each clause has at most three dissimilar neighbours and each 
variable appears in at most three dissimilar clauses, and no two
dissimilar clauses have two common variables.

Then, we can select $k \leq m(1/6+\epsilon_m)$ variables, such that
branching on all possible values for all of these variables, and then
doing simplification based on repeated use of Case 1.(i) to 1.(iv)
gives two groups of clauses, each having three literals,
where the two groups have no common
variables, and
\begin{enumerate}[\rm(a)]
\item each clause in each group has at most three dissimilar neighbours,
\item each variable appears in at most three dissimilar clauses,
\item no pair of dissimilar clauses have two common variables,
\item the number of dissimilar clauses in each group is at most $(m-k+2)/2$.
\end{enumerate}
\end{propo}

\begin{proof}
To prove the proposition, consider each dissimilar clause 
as a vertex, with edge connecting two dissimilar clauses if they
have a common variable. 
Using the bisection width result \cite{Gas10,GS17,MP06}, one
can partition the dissimilar clauses into two groups (differing
by at most one in cardinality) such that
there exist at most $k \leq (1/6+\epsilon_m)\times m$ edges between the two
groups, that is there are at most $(1/6+\epsilon_m) \times m$ common
variables between the two groups of clauses. 
One can assume without loss of generality
that at most one clause has all its neighbours on the other
side. This holds as if there are two dissimilar clauses, say one in each half,
which have all their neighbours on the other side, then we can switch
these two clauses to the other side and decrease the size of the cut.
On the other hand, if both these clauses (say $A$ and $B$) 
belong to the same side, then
we can switch $A$ to the other side, and switch the side of
one of $B$'s neighbours --- this also decreases the size of the cut.

To see that the properties mentioned ((a), (b) and (c)) are preserved, 
suppose in a clause $(x,y,z)$, we branch on $x$ and thus link $y$ with
$z$; here we assume without loss of generality that
$x,y,z$ are all positive literals. 
Note that as $(x,y,z)$ has at most three neighbours, one of which
contains $x$, there can be at most two other neighbours of the
clause $(x,y,z)$ which contain $y$ or $z$.

First suppose $y$ (respectively $z$)
does not appear in any other clause. Without loss of generality
assume that $y$ gets dropped and replaced by $z$ or $\neg z$ based
on the linking.
Then dropping the clause $(x,y,z)$ and
replacing $y$ by $z$ does not increase the number of dissimilar clauses 
that $z$ appears in, nor does it increase the number
of neighbours of these clauses as there is no change in variable name 
in any clause which is not dropped. 

Next suppose both $y$ and $z$ appear in exactly one other dissimilar clause,
say $(y',a,b)$ and $(z',c,d)$, where $y'$ and $z'$ are literals involving
$y$ and $z$ respectively.
In that case, linking $y$ and $z$ (and replacing $z$ by $y$), makes these
two clauses neighbours (if not already so) --- which is compensated by 
the dropping of the neighbour $(x,y,z)$; the number of
clauses in which $y$ appears remains two.
In case these two clauses were already neighbours (say $a=c$ or $\neg c$), then 
due to application of Case 1.(iv), 
$b$ and $d$ get linked, clauses $(y,a,b)$ and $(z,c,d)$ thus
become similar (resulting in decrease in the neighbour by one for
these clauses) and the above analysis can then be recursively applied
for linking $b$ with $d$.

Now considering the edges (and corresponding common variable for the
edge) in the cut,
and branching on all these variables (while being consistent with
$s_1$ and $s_2$) and then doing simplification as
in Cases 1(i) to 1(iv), we have that
each partition is left with 
at most $(m+1-(k-1))/2$ dissimilar clauses. This holds as,
by our assumption above, except
maybe for one clause, all dissimilar clauses have at most two neighbours on 
the other side.
Thus, by linking the remaining variables for each of the clauses involved
in the cut, we can remove $(k-1)/2$ dissimilar clauses on each side
using Case 1(iii).
\end{proof}

\noindent
Thus, one can recursively apply the above modifications in Case 2
to each of the two groups of clauses, one after other, until all 
the variables have been assigned the values or linked to other variables
(where the leaf cases occur when the number
of dissimilar clauses is small enough to use brute force assigning values to
all of the variables).

Now we count how many variables need to be branched for Case 2 in total
if one starts with $m$ clauses involving $n$ variables.
The worst case happens when $k=(1/6+\epsilon_m)m$ and the total
number of variables which need to be branched on is
$m(1+5/12+5^2/(12^2)+\ldots)*(1/6+\epsilon)$, where one can take $\epsilon$
as small as desired for corresponding large enough $m$.
Thus the number of variables branching would be
$m(2/7+12\epsilon/7) \leq n(4/21+24\epsilon/21)$.
As branching on each variable gives at most $4$ children, the number
of leaves (and thus complexity of the algorithm based on
Case 2) is bounded by $4^{4n/21+o(n)}$.

\subsection{Overall Complexity of the Algorithm}

\noindent
Note that modifications in each of the above cases takes
polynomial time in the original formula $\phi$.

Visualize the running of the above algorithm as a search tree, where
the root of the tree is labeled as the starting
problem $\MHD(\phi,\phi,V=X,s_1=\emptyset,s_2=\emptyset,P)$,
with $P$ having $p_{main}=1$, $p_{x,i,j}=q_{x,i,j}$.

At any node, if a simplification case applies, then the node 
has only one child with the corresponding updated parameters. If a braching
case applies, then the node has children corresponding to the parameters
in the branching.

As the work done at each node is polynomial in the length of $\phi$,
the overall time complexity of the algorithm is $poly(n,|\phi|)\times $
(number of leaves in the above search tree).

We thus analyze the number of possible leaves the search tree
would generate.

Suppose $T(r)$ denotes the number of leaves rooted at a node
$\MHD(\ldots,V,\ldots)$, where $V$ has $r$ variables.

Case 1.(i) to Case 1.(iv) and Case 1.(vi).1 do not involve any branching.

If Case 1.(v) is applied to a $\MHD$ problem involving $r$ variables, then
it creates at most four subproblems, each having at most $r-5$ variables.
Thus, the number of leaves generated in this case is
bounded by $4 T(r-5)$.
Note that $T(r) = O(\alpha^r)$, for $\alpha \geq \alpha_0=1.3196$
satisfies the constraints of this equation.

If Case 1.(vi).2 is applied to a $\MHD$ problem involving $r$ variables, then it
creates at most 4 subproblems each involving at most $r-5$ variables.
Thus, the number of leaves generated in this case is
bounded by $4 T(r-5)$.
Note that $T(r) = O(\alpha^r)$, for $\alpha \geq \alpha_0=1.3196$
satisfies the constraints of this equation.

If Case 1.(vi).3 is applied to a $\MHD$ problem involving $r$ variables, then it
creates at most 9 subproblems each involving at most $r-8$ variables.
Thus, the number of leaves generated in this case is
bounded by $9 T(r-8)$.
Note that any $T(r) = O(\alpha^r)$, for $\alpha \geq \alpha_0=1.3162$
satisfies the constraints of this equation.

If Case 1.(vii) is applied to a $\MHD$ problem involving $r$ variables, then it
creates at most 6 subproblems, one involving at most $r-4$ variables and
the other involving at most $r-7$ variables.
Thus, the number of leaves generated in this case is
bounded by $T(r-4)+5T(r-7)$.
Note that any $T(r) = O(\alpha^r)$, for $\alpha \geq \alpha_0=1.3298$
satisfies the constraints of this equation.

If Case 2 is applied to a $\MHD$ problem of $r$ variables, then it creates 
a search tree which contains at most $O(4^{4r/21+o(r)})$ leaves.
Note that any $T(r) = O(\alpha^r)$, for $\alpha \geq \alpha_0=1.3023$
satisfies the constraints of this equation.

Thus, the formula $T(r)= O(1.3298^r)$ bounds the number of leaves
generated in each of the cases above, for large enough $r$. 
Thus, we have the theorem:

\begin{theorem}
Given a 3XSAT formula $\phi$, one can find in time 
$O(poly(n,|\phi|) \times 1.3298^n)$ 
the maximum hamming distance between any two satisfying assignments for
$\phi$.
\end{theorem}

\ifx\versy\webpage
\section{Comparing with Reductions to Known Methods}\label{sec:comparision}

\noindent
An early approach of computing maximal Hamming distances between solutions
was an algorithm which (a) enumerates all the solutions on one side
and then (b) finds for each solution of (a) the most distant solution on
the other side. This method exploited that for (b), one can use a method
of maximising a variable weighted X3SAT which by Porschen and
Plagge \cite{PP10} takes approximately time $1.1192^n$; one cannot
say at least, as they did not prove a lower bound for this but only
an upper bound. The performance of this
algorithmic idea mainly depends on the number of solutions in (a). If one
considers $n = 2n'+1$ clauses $(x_1, x_{2m}, x_{2m+1})$ with
$m=1,\ldots,n'$, then there are $1+2^{n'}$ solutions which is, for
most $n$, at least $1.414^n$. So the overall runtime is approximately
$1.5825^n$.

In the following, we want to lay out more in detail why
referring to standard methods like Max 2-CSP or the above
mentioned algorithm does not give
better bounds than the algorithm provided in the current paper.
The following three remarks, the first for the special case
and the next two for the full problem, give some estimated bounds on these
type of approaches. The goal is to try to give a fair comparison based on a
reasonable way of using this approach.

\begin{remy} \label{rem:csp}
Assume that one has to compute the maximum Hamming distance for an
X3SAT formula which meets the specification of Case 2 in the algorithm.
Then one could formalise this as a Max 2-CSP problem as follows:
One makes a graph of all clauses where a clause is considered to be
a set of $3$ nodes. Each node takes a colour from $\{1,2,3\}^2$ where
the coordinate $1,2,3$ indicates whether the first, second or third literal
is made true in the clause, as it is X3SAT, exactly one of these three options
applies. Furthermore, the two coordinates in the pair
refer to the first and the second solution of the X3SAT problem.
If two clauses share a variable, they are neighbours; for neighbours
one makes the hard constraint that the shared variables in the two
nodes are given consistent values. The weak constraint is the Hamming
distance of the two solutions, here one evaluates for each variable the
Hamming distance between the two solutions in the first node where this
variable occurs; thus the distance between the two solutions in the node
is a number from $0$ to $3$ which reflects the Hamming distance of the
variables in the solution which occur in this node first. The hard constraints
have for each pair the weight $n+1$ (greater than the sum of all weak
constraints) in the case that the hard constraint is satisfied and $0$
in the case that the hard constraint is not satisfied. Now, as the
underlying graph has degree at most $3$, the algorithm of Gaspers and
Sorkin \cite{GS17} provides the bound of $9^{\frac{m}{5}+o(m)}$ where
$m$ is the number of nodes in the graph, that is, the number of clauses
in the given formula. This number $m$ is at most $\frac{2}{3} \times n$
due to the special form of the graph. Thus the overall time complexity
is $9^{\frac{2n}{15}+o(n)}$ and this is contained in $O(1.3404^n)$.
Thus the complexity of the na\"ive invocation of Max 2-CSP in an
important special case would give a bound worse than the algorithm
for solving the full problem in this paper; so it pays off to make
a specialised algorithm for the problem of determining
the maximum Hamming distance of two solutions of a X3SAT instance.
\end{remy}

\begin{remy} \label{rem:appendix}
One could also try to solve the full problem with the invocation of the
Max 2-CSP algorithm. Gaspers and Sorkin~\cite{GS17} provide the bound of
$9^{9h/50}$ for a instance with $h$ edges in the underlying graph (which
optimises a sum over the value functions along the edges of the graph
plus a further sum over the value functions of the nodes). Again one would
take the nodes as clauses and for every variable occurring in $k+1$ clauses,
one would make $k$ edges, connecting the first and second clause where it
occurs, the second and third clause where it occurs, $\ldots$, the
$k$-th and $k+1$-st clause where it occurs. Along these edges, one puts
the hard constraint that the corresponding values of the variable are
the same and they carry the weight $n+1$; along the nodes one puts the
weak constraint equal to the Hamming distances of the value-vectors of
those variables in the two solutions which occur in this clause
but do not appear in any earlier clause. The maximum constraint value
can, if there is a solution, only be taken by a pair of solutions which
satisfies all the hard constraints and its value is $n+1$ times the number
of hard constraints plus the maximum Hamming distance.

Let $r$ be the average number of clauses in which
a variable occurs. Now the
time bound for the PSPACE algorithm for this problem is
$9^{9/50 \times (rn-n+o(n))}$ in dependence
of $r$ and $n$. For $r = 3.0$ this is contained in $O(2.2057^n)$, for $r = 2.3$
this is contained in $O(1.6723^n)$ and for $r = 2.0$ this is contained in
$O(1.4852^n)$. For the case that one does not want to specify an average
degree, the time bound can be estimated by $O(9^{9/50 \times (3m-n)+o(m)})$
as every clause has at most $3$ edges connecting to later clauses and
for each variable which occurs the last time in a clause, there is no
edge connecting to a further later clause, thus the $-n$ term.
In the case that one uses exponential space algorithms, there are
slightly better bounds supplied by Scott and Sorkin \cite{SS07}
which are, for the above case, $9^{(13/75+o(1))\times (r-1)\times n}$
giving $O(1.4636^n)$ for $r=2.0$, $O(1.6407^n)$ for $r=2.3$ and
$O(2.1420^n)$ for $r=3.0$.
These algorithms are, even for the moderate value $r=3.0$,
not competitive with Dahll\"of's original algorithm \cite{Dah05,Dah06}
which solves the maximal Hamming distance even for XSAT and not only X3SAT.
Thus it pays off to make a specialised taylormade algorithm rather than
to plug in a known general method for the case of Max Hamming Distance X3SAT.
\end{remy}

\begin{remy}
One might ask why the above approach takes the clauses as vertices
of the CSP graph and not the variables. The main reason is that published
results which give the CSP complexity in terms of vertices are, mainly,
just the paper of Williams \cite{Wil07}: He shows that
one can solve the Max 2-CSP constraint problem in time $O^*(1.732^n)$;
however, this result is for binary variables only. Williams' method does
not allow that a variable takes four values, as otherwise we
could compress pairs of variables into one variable and bring down
the complexity to $O^*(1.732^{n/2})$ and then do the same thing again.
For that reason, to cast Max Hamming X3SAT into this framework, we would
have to go for an $O^*(1.732^{2n})$ algorithm, which is much slower than
Dahll\"of's algorithm \cite{Dah05,Dah06} of $O(1.8348^n)$.

A better way would be to use
Max weighted $2$SAT with $3n$ variables instead of $n$.
Wahlstr\"om \cite{Wah07} provides for this an algorithm in time
$O^*(1.2377^n)$; this algorithm would then be placed on instances
with $3n$ variables which corresponds to $O(1.8961^n)$ which is
slightly above Dahll\"of's algorithm. For this, recall that maximum
weighted $2$SAT assigns to every variable a weight and searches
for a solution of the $2$SAT formula which maximises the weight.
The translation is as follows:
We let $(x_1,\ldots,x_n)$ and $(y_1,\ldots,y_n)$ represent the
two solutions of X3SAT problem; $(z_1,\ldots,z_n)$ is an auxiliary vector with
the constraint that $z_k$ can only be $1$ when $x_k=0$ and $y_k=1$;
this is achieved by putting the clauses $\neg z_k \vee \neg x_k$ and
$\neg z_k \vee y_k$ into the $2$SAT formula.
Furthermore, for each X3SAT clause, we put into the $2$SAT
formula the conditions that no two literals in the clause are satisfied
at the same time, that is, for a clause $x_i \vee x_j \vee \neg x_k$
in the X$3$SAT instance the
corresponding $2$SAT clauses would be $\neg x_i \vee \neg x_j$,
$\neg x_i \vee x_k$ and $\neg x_j \vee x_k$.
Similarly we put the $2$SAT clauses for the $y$-variables.
Now let $i_k$ be the number of clauses containing $x_k$ in the
X$3$SAT instance and $j_k$ be the number of clauses containing $\neg x_k$
in the X$3$SAT instance.
The weights are as follows:
\begin{enumerate}
\item $x_k=1$ has weight $(n+1) \cdot i_k+1$;
\item $x_k=0$ has weight $(n+1) \cdot j_k$;
\item $y_k=1$ has weight $(n+1) \cdot i_k$;
\item $y_k=0$ has weight $(n+1) \cdot j_k+1$;
\item $z_k=1$ has weight $2$;
\item $z_k=0$ has weight $0$.
\end{enumerate}
Note that the Hamming distance is not a linear function and therefore
we need the variable $z_k$
to correct errors when computing the Hamming distance from weights
on $x_k$ and $y_k$; this correction is good when the number of
$z_k=1$ is maximised, given assignments of
$x_1,\ldots,x_n$ and $y_1,\ldots,y_n$ which solve the X3SAT problem.
Now the overall maximal weight of two solutions $x,y$ and $z$ being
maximised is
$$
   2 \cdot m \cdot (n+1) + n + HD(x,y),
$$
where $m$ is the number of clauses in the X3SAT instance.
The reason for this is that every satisfied clause for $x$ contributes
one time $n+1$ and every satisfied clause for $y$ also contributes
one time $n+1$ to the sum.
Additionally, the part of the weights which is not determined by X3SAT
clauses is, for each $k$, equal to $x_k + 1-y_k+2z_k$; the part
$x_k+1-y_k$ is $1$ if $x_k$ and $y_k$ are the same;
$2$ if $x_k=1$ and $y_k=0$ and $0$ if $x_k=0$ and $y_k=1$.
Note that $z_k$ can be $1$ iff $x_k=0$ and $y_k=1$. Therefore,
by chosing $z_k$ suitably, if $x_k$ and $y_k$ differ then the sum of 
the parts of the weight of $x_k,y_k,z_k$ not linked to coding X3SAT
clauses is $2$ else this sum is $1$. Thus taking $x_k,y_k$ and the maximised
$z_k$ into account, the additional weight of $x_k+1-y_k+2z_k$ is
$1+HD(x_k,y_k)$. Taking this for all $n$
variables into account gives $n+HD(x,y)$.

Alternatively, one could use the approach of Porschen and
Plagge \cite{PP10} to solve variable-weighted X3SAT in
time $O(2^{0.16255 \times n})$.
However, X3SAT formulas are not very suitable for coding the Hamming distance
and therefore three additional variables are needed for each pair of $x_k$
and $y_k$, giving a total of $5n$ variables. This blown up problem then has
the performance of $O(1.7566^n)$, so the overall performance is better
than Dahll\"of's algorithm but worse than the one of
Fu, Zhou and Yin \cite{FZY12}.
\end{remy}
\fi

\section{Conclusion and Future Work}

\noindent
In this paper, we considered a branching algorithm to compute the Max
Hamming Distance X3SAT in $O(1.3298^n)$ time. Our novelty lies
in the preservation of structure at both sides of the formula while we branch.

Our method is faster than the na\"ive invocation of the Max 2-CSP algorithm
\ifx\versy\conf
(see the discussion in the second-last section of the technical report
version on http://www.arxiv.org of this paper).
\fi
\ifx\versy\webpage
(see the discussion in the previous section).
\fi
Even if one assumes that every clause has only three neighbours (as in Case 2,
but now from the start), the usage of the Max 2-CSP algorithm results in a
run-time of $9^{2/15 \times n+o(n)}$ which is contained in $O(1.3404^n)$.
Without this assumption, the na\"ive invocation of the Max 2-CSP algorithm
is much worse. Also other invocations of known methods do not give good
timebounds.

Our time bound of $O(1.3298^n)$ is 
achieved by using simple analysis to analyse our branching rules.
Our algorithm uses only polynomial space during its computations.
This can be seen from the fact that the recursive calls at the branchings
are independent and can be sequentialised; each calling instance therefore
needs only to store the local data; thus each node of the call tree uses
only $h(n)$ space for some polynomial $h$. The depth of the tree is
at most $n$ as each branching reduces the variables by $1$;
thus the overall space is at most $h(n) \times n$ space. 

Furthermore, as we determine the number of pairs of solutions
with Hamming distance $k$ for $k=0,1,\ldots,n$, where $n$ is
the number of variables, one might ask
whether this comes with every good algorithm for free or whether there
are faster algorithms in the case that one computes merely the
maximum Hamming distance of two solutions.

\end{document}